\documentclass[a4paper,USenglish]{article}

\usepackage[utf8]{inputenc}
\usepackage[USenglish]{babel}
\usepackage{booktabs}
\usepackage{paralist}
\usepackage{microtype}
\usepackage{dsfont}
\usepackage{amssymb}
\usepackage{mathtools}
\usepackage{xfrac}
\usepackage{tabularx}
\usepackage{a4wide}
\usepackage{multicol}
\usepackage{graphicx}
\usepackage{ifthen}

\makeatletter
\renewcommand\paragraph{\@startsection{paragraph}{4}{\z@}%
                                    {3.25ex \@plus1ex \@minus.2ex}%
                                    {-1em}%
                                    {\normalfont\normalsize\bfseries}}

\renewcommand\subparagraph{\@startsection{subparagraph}{5}{\parindent}%
                                       {3.25ex \@plus1ex \@minus .2ex}%
                                       {-1em}%
                                      {\normalfont\normalsize\itshape}}
\makeatother

\usepackage{tikz}
\usetikzlibrary{decorations.pathmorphing}
\usetikzlibrary{trees}
\usetikzlibrary{calc,backgrounds}
\usetikzlibrary{arrows.meta}
\tikzset{
    none/.style={color=white},
    smallvertex/.style={circle,draw,fill=black, inner sep= .75pt},
    vertex/.style={circle,draw,fill=black, inner sep= 1.5pt},
    marked/.style={draw,color=black,inner sep= 3.1pt},
    tinycircle/.style={draw,circle,fill=white,inner sep=1.10pt},
    gray/.style={fill=black!30},
    hollow/.style={circle,draw,inner sep=1pt},
    node/.style={color=black,circle,draw,dotted,very thick, inner sep= 1.5pt, },
    arc/.style={->,> = latex', thick},
    dashedarc/.style={->,> = latex', dashed, line width=.75pt},
    dashdotarc/.style={->,> = latex', dash dot, line width=.75pt},
    bluearc/.style={preaction={draw,blue,-,double=blue,double distance=.75pt}},
    orangearc/.style={preaction={draw,orange,-,double=orange,double distance=.75pt}},
    edge/.style={-,line width=.75pt},
}
\pgfdeclarelayer{bg}
\pgfsetlayers{bg,main}

\usepackage{amsthm}

\usepackage[pagebackref,pdfdisplaydoctitle,menucolor=orange!40!black,filecolor=magenta!40!black,urlcolor=blue!40!black,linkcolor=red!40!black,citecolor=green!40!black,colorlinks]{hyperref}

\usepackage[nameinlink,sort&compress]{cleveref}
\crefname{rrule}{Reduction Rule}{Reduction Rules}
\crefname{construction}{Construction}{Constructions}
\crefname{claim}{Claim}{Claims}
\crefname{paragraph}{Paragraph}{Paragraphs}
\crefname{observation}{Observation}{Observations}
\crefname{theorem}{Theorem}{Theorems}
\crefname{lemma}{Lemma}{Lemmata}
\crefname{proposition}{Proposition}{Propositions}
\crefname{corollary}{Corollary}{Corollaries}
\crefname{remark}{Remark}{Remarks}
\crefname{section}{Section}{sections}
\crefname{chapter}{Chapter}{Chapters}
\crefname{figure}{Figure}{Figures}
\crefname{table}{Table}{Tables}
\crefname{definition}{Definition}{Definitions}
\crefname{algorithm}{Algorithm}{Algorithms}
\crefname{equation}{equation}{equations}
\crefname{constraint}{constraint}{constraints}
\crefname{ineq}{inequality}{inequalities}
\creflabelformat{constraint}{#2(#1)#3}
\theoremstyle{plain}
\newtheorem{theorem}{Theorem}
\newtheorem{lemma}[theorem]{Lemma}
\newtheorem{corollary}[theorem]{Corollary}
\newtheorem{observation}[theorem]{Observation}

\newtheorem{rrule}{Reduction Rule}
\theoremstyle{definition}

\theoremstyle{remark}
\newtheorem*{remark}{Remark}

\usepackage[square,sectionbib,numbers,sort&compress]{natbib}
\usepackage{bibentry}
\usepackage{xcolor}
\definecolor{myred}{RGB}{255,153,153}
\definecolor{mygreen}{RGB}{153,255,153}
\DeclareMathOperator{\tw}{\mathrm{tw}}
\DeclareMathOperator{\pw}{\mathrm{pw}}
\DeclareMathOperator{\td}{\mathrm{td}}
\DeclareMathOperator{\cw}{\mathrm{cw}}
\DeclareMathOperator{\mw}{\mathrm{mw}}
\DeclareMathOperator{\fen}{\mathrm{fen}}
\DeclareMathOperator{\fvn}{\mathrm{fvn}}
\DeclareMathOperator{\diam}{diam}
\DeclareMathOperator{\Wone}{W[1]}
\DeclareMathOperator{\Wtwo}{W[2]}
\DeclareMathOperator{\NP}{NP}

\DeclareMathOperator{\Adj}{Adj}
\DeclareMathOperator{\spath}{Path}
\DeclareMathOperator{\idist}{Dist}
\DeclareMathOperator{\visit}{Visit}
\DeclareMathOperator{\mso}{MSO}

\newcommand{\N}{\mathds{N}}

\newcommand{\problemdef}[3]{
	\begin{center}
		\begin{minipage}{0.93\textwidth}
			\textsc{#1}\\[2mm]
			\setlength{\tabcolsep}{3pt}
			\begin{tabularx}{\textwidth}{@{}lX@{}}
				\normalsize \textbf{Input:} 	& \normalsize #2 \\
				\normalsize \textbf{Question:} 	& \normalsize #3
			\end{tabularx}
		\end{minipage}
	\end{center}
}

\title{\LARGE \bf Parameterized Complexity of Geodetic Set}

\author{Leon Kellerhals \and Tomohiro Koana\footnote{TK was partially supported by the DFG projects FPTinP (NI 369/16) and MATE (NI 369/17)}}

\date{Technische Universität Berlin, Algorithmics and Computational Complexity\\\texttt{\small \{leon.kellerhals,tomohiro.koana\}@tu-berlin.de}}

\begin{document}

\maketitle

\begin{abstract}
	A vertex set~$S$ of a graph~$G$ is \emph{geodetic} if every vertex of~$G$ lies on a shortest path between two vertices in~$S$.
	Given a graph $G$ and $k \in \N$, the~$\NP$-hard \textsc{Geodetic Set} problem asks whether there is a geodetic set of size at most $k$.
	Complementing various works on \textsc{Geodetic Set} restricted to special graph classes, we initiate a parameterized complexity study of \textsc{Geodetic Set} and show, on the negative side, that \textsc{Geodetic Set} is~$\Wone$-hard when parameterized by feedback vertex number, path-width, and solution size, combined.
	On the positive side, we develop fixed-parameter algorithms with respect to the feedback edge number, the tree-depth, and the modular-width of the input graph.
\end{abstract}

\section{Introduction}
Let~$G$ be an undirected, simple graph with vertex set~$V(G)$ and edge set~$E(G)$.
The \emph{interval}~$I[u, v]$ of two vertices~$u$ and~$v$ of~$G$ is the set of vertices of~$G$ that are contained in any shortest path between~$u$ and~$v$.
In particular, $u, v \in I[u, v]$.
For a set~$S$ of vertices, let~$I[S]$ be the union of the intervals~$I[u, v]$ over all pairs of vertices~$u$ and~$v$ in~$S$.
A set of vertices~$S$ is called \emph{geodetic} if~$I[S]$ contains all vertices of~$G$.
In this work we study the following problem (see an exemplary illustration in \cref{fig:example}):
\problemdef{Geodetic Set}
{A graph~$G$ and an integer~$k$.}
{Does~$G$ have a geodetic set of cardinality at most~$k$?}

\begin{figure}[t]
	\centering
	\begin{tikzpicture}
		\node[hollow,minimum width=10pt,fill=gray]	(v1) at (-1,2)	 {};
		\node[hollow,minimum width=10pt,fill=gray]	(v2) at (-1,0)	 {};
		\node[hollow,minimum width=10pt]		(v3) at (1,2)	 {};
		\node[hollow,minimum width=10pt]		(v4) at (1,0)	 {};
		\node[hollow,minimum width=10pt]		(v5) at (3,2)	 {};
		\node[hollow,minimum width=10pt]		(v6) at (3,0)	 {};
		\node[hollow,minimum width=10pt]		(v7) at (2,1)	 {};
		\node[hollow,minimum width=10pt]		(v8) at (5.5,1)	 {};
		\node[hollow,minimum width=10pt,fill=gray]	(v9) at (4.25,.5){};

		\draw[edge] (v1) -- (v3) -- (v5) -- (v8) -- (v9) -- (v6) -- (v7) -- (v3) -- (v4) -- (v6);
		\draw[edge] (v2) -- (v4) -- (v7);
		\draw[edge] (v5) -- (v6);

	\end{tikzpicture}
	\caption{An exemplary graph. The gray vertices form a minimum geodetic set. The shortest paths between the top left and the bottom right gray vertex cover all vertices except for the bottom left vertex. Observe that every geodetic set contains all degree-one vertices.}
	\label{fig:example}
\end{figure}
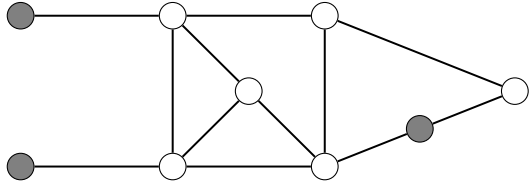

Atici \cite{Ati02} showed that \textsc{Geodetic Set} is~$\NP$-complete on general graphs, and it was shown that the hardness holds even if the graph is planar \cite{CFGGR20}, subcubic \cite{BPPRRS18}, chordal, or bipartite chordal \cite{DPRS10}.
Although not stated, $\Wtwo$-hardness for the solution size $k$ directly follows from the reduction for the latter result of Dourado et al. \cite{DPRS10}.
On the positive side, the problem was shown to be polynomial-time solvable for cographs, split graphs and unit interval graphs~\cite{DPRS10}.
Also, upper bounds on the geodetic set size in Cartesian product graphs were studied~\cite{BKH08}.

For a graph $G$ and $k \in \N$, the closely related \textsc{Geodetic Hull} problem asks whether there is a vertex set $S \subseteq V(G)$ with $I^{|V(G)|}[S] = V(G)$ and $|S| \le k$, where $I^0[S]= S$ and $I^j[S] = I[I^{j - 1}[S]]$ for $j > 0$.
\textsc{Geodetic Hull} is~$\NP$-hard on bipartite~\cite{AMSSW16}, chordal~\cite{BDPR18}, and~$P_9$-free graphs~\cite{DPR16}.
Recently, Kant\'{e} et al.~\cite{KMS19} studied the \emph{parameterized complexity} of \textsc{Geodetic Hull}: they proved that the problem is~$\Wtwo$-hard when parameterized by~$k$, and~$\Wone$-hard but in XP when parameterized by tree-width.\footnote{Informally, this means it can be solved in polynomial time for graphs of constant tree-width.}

\paragraph{Our Contributions.}
Comparing the algorithmic complexity of \textsc{Geodetic Hull} and \textsc{Geodetic Set}, one can observe that both problems are trivial on trees (take all leaves into the solution).
But while \textsc{Geodetic Hull} is polynomial-time solvable on graphs of constant tree-width, the complexity of \textsc{Geodetic Set} on graphs of tree-width two is unknown to the best of our knowledge.
Motivated by this gap, we study the parameterized complexity of \textsc{Geodetic Set} for structural parameters such as tree-width that measure the tree-likeness of the input graph, providing both positive and negative results.

We start off by showing that \textsc{Geodetic Set} is~$\Wone$-hard with respect to tree-width.
More specifically, we show that \textsc{Geodetic Set} is~$\Wone$-hard for feedback vertex number, path-width, and solution size, all three combined (\cref{sec:fvn}), using a parameterized reduction from the~$\Wone$-hard \textsc{Grid Tiling} problem \cite{Marx07}.
Since this reduction implies~$\NP$-hardness, this complements previous results by providing a more fine-grained view on computational tractability in terms of parameterized complexity instead of studying special graph classes.

We complement the~$\Wone$-hardness by presenting two fixed-parameter tractability results for \textsc{Geodetic Set}.
First, we show that \textsc{Geodetic Set} is fixed-parameter tractable with respect to the feedback edge number (\cref{sec:fen}).
It turns out to be quite effortful to obtain fixed-parameter tractability,
requiring the design and analysis of polynomial-time data reduction rules and branching before employing the main technical trick: Integer Linear Programming (ILP) with a bounded number of variables.
To the best of our knowledge, this is the first usage of ILP when solving \textsc{Geodetic Set}.

Second, we show that \textsc{Geodetic Set} is fixed-parameter tractable with respect to clique-width combined with diameter (\cref{sec:cw}); note that \textsc{Geodetic Set} is~$\NP$-hard even on graphs with constant diameter \cite{DPRS10}, and~$\Wone$-hard with respect to clique-width (this follows from our first result).
Our result exploits the fact that we can express \textsc{Geodetic Set} in an~$\mso_1$ logic formula, the length of which is upper-bounded in a function of the diameter of the graph.
A direct consequence of this result is that \textsc{Geodetic Set} is fixed-parameter tractable with respect to tree-depth and with respect to modular-width.

\cref{fig:params} gives an overview of the parameters for which we obtain positive and negative results, and presents their interdependence.

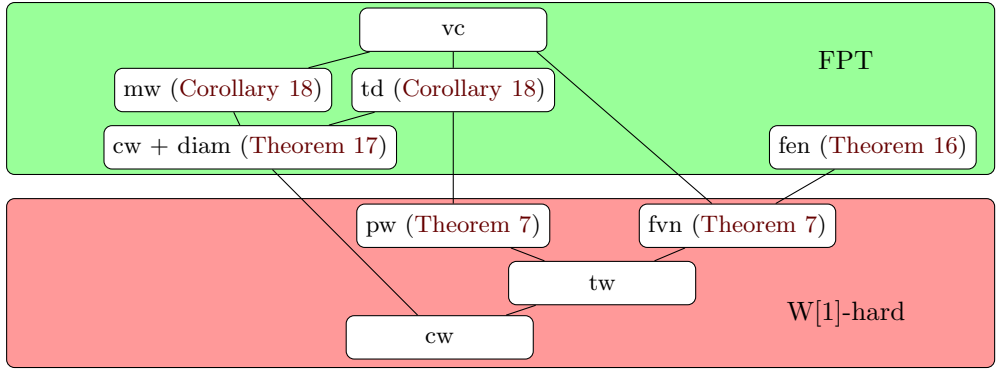
\begin{figure}[t]
	\centering
	\begin{tikzpicture}[xscale=3.56, yscale=0.8]
		\tikzset{
			param/.style={draw, fill=white, rectangle, rounded corners=3, font=\small, minimum width=7em, minimum height=3.8ex},
		}
		\draw[draw=black, fill=mygreen, rounded corners=3] (-1.6, 5.65) rectangle (2.05, 2.8);
		\draw[draw=black, fill=myred, rounded corners=3] (-1.6, 2.4) rectangle (2.05, -0.4);

		\node[param] at (0, 0.1) (cw) {cw};
		\node[param] at (0.6, 1) (tw) {tw};
		\node[param] at (0.05, 1.95) (pw) {pw (\Cref{thm:w1-fvn})};
		\node[param] at (1.1, 1.95) (fvn) {fvn (\Cref{thm:w1-fvn})};
		\node at (1.5, 0.5) {W[1]-hard};

		\begin{scope}[shift={(0, 0.2)}]
			\node[param] at (1.6, 3.05) (fen) {fen (\Cref{thm:fen})};
			\node[param] at (-0.7, 3.05) (cwd) {cw + diam (\Cref{cor:twdiam})};
			\node[param] at (0.05, 4) (td) {td (\Cref{thm:tdtd})};
			\node[param] at (-0.8, 4) (mw) {mw (\Cref{thm:tdtd})};
			\node[param] at (0.05, 5.0) (vc) {vc};
			\node at (1.5, 4.5) {FPT};
		\end{scope}
		\draw (cw) -- (tw) -- (pw);
		\draw (cw) -- (cwd) -- (td) -- (vc);
		\draw (cwd) -- (mw) -- (vc.195);
		\draw (tw) -- (fvn) -- (fen);
		\draw (pw) -- (td);
		\draw (fvn) -- (vc.345);
	\end{tikzpicture}
	\caption{
		An overview of our results for \textsc{Geodetic Set}, containing the parameters vertex cover number (vc), modular-width (mw), tree-depth (td), clique-width (cw), diameter (diam), feedback edge number (fen), path-width (pw), feedback vertex number (fvn) and tree-width (tw).
		An edge between two parameters indicates that the one below is smaller than some function of the other.
	}
	\label{fig:params}
\end{figure}

\section{Preliminaries}
\label{sec:preliminaries}
For~$n \in \N$ let~$[n] = \{1, 2, \dots, n\}$.
The \emph{distance}~$d_G(u, v)$ between two vertices~$u$ and~$v$ in~$G$ is the length of a shortest path between~$u$ and~$v$ (also called \emph{shortest~$u$--$v$-path}).
We drop the subscript~$\cdot_G$ if~$G$ is clear from context.
Note that~$w$ belongs to~$I[u, v]$ if and only if~$d_G(u, v) = d_G(u, w) + d_G(w, v)$.
The \emph{diameter}~$\diam(G)$ of~$G$ is the maximum distance between any two vertices of~$G$.
A \emph{multigraph}~$G$ consists of a vertex set and an edge multiset.
Note that in a multigraph, we count self-loops twice for the vertex degree.

A set~$F \subseteq E(G)$ is a \emph{feedback edge set} if~$G \setminus F$ is a forest.
The \emph{feedback edge number}~$\fen(G)$ is the size of a smallest such set.
Analogously, a set~$V' \subseteq V(G)$ is a \emph{feedback vertex set} if~$G - V'$ is a forest.
The \emph{feedback vertex number}~$\fvn(G)$ is the size of a smallest such set.

For a graph~$G$, a \emph{tree decomposition} is a pair~$(T, B)$, where~$T$ is a tree and~$B \colon V(T) \to 2^{V(G)}$ such that
(i) for each edge~$uv \in E(G)$ there exists~$x \in V(T)$ with~$u, v \in B(x)$, and 
(ii)~for each~$v \in V(G)$ the set of nodes~$x \in V(T)$ with $v \in B(x)$ forms a nonempty, connected subtree in~$T$.
The \emph{width} of~$(T, B)$ is~$\max_{x \in V(T)}(|B(x)| - 1)$.
The \emph{tree-width}~$\tw(G)$ of~$G$ is the minimum width of all tree decompositions of~$G$.
The \emph{path-width}~$\pw(G)$ of~$G$ is the minimum width of all tree decompositions~$(T, B)$ of~$G$ for which~$T$ is a path.

The \emph{tree-depth} of a connected graph~$G$ is defined as follows~\cite{NM12}.
Let~$T$ be a rooted tree with vertex set~$V(G)$, such that if~$xy \in E(G)$, then~$x$ is either an ancestor or a descendant of~$y$ in~$T$.
We say that~$G$ is \emph{embedded in~$T$}.
The \emph{depth} of~$T$ is the number of vertices in a longest path in~$T$ from the root to a leaf.
The \emph{tree-depth}~$\td(G)$ of~$G$ is the minimum~$t$ such that there is a rooted tree of depth~$t$ in which~$G$ is embedded.

We next define the \emph{modular-width} of a graph~$G$~\cite{HP10}.
A vertex set~$M \subseteq V(G)$ is a \emph{module} if for all~$u, v \in M$ it holds that~$N(v) \cap V(G) \setminus M = N(w) \cap V(G) \setminus M$.
We call a module~$M$ \emph{trivial}, if~$|M| \le 1$ or~$M = V$, and we call it \emph{strong} if for every other module~$M'$ of~$G$ we have that~$M \cap M' = \emptyset$, or that one is a subset of the other.
A graph that only admits trivial modules is called \emph{prime}.
Every non-singleton graph can be uniquely partitioned into maximal strong modules~$\mathcal P = \{M_1, \dots, M_\ell\}$ with~$\ell \ge 2$.
Recursively partitioning the graphs~$G[M_i]$ in this way until every module is a single vertex yields a \emph{modular decomposition} of~$G$.
The modular-width is the largest number of trivial modules in a \emph{prime subgraph} $G[M_i]$ of the modular decomposition of~$G$.

A \emph{parameterized problem} is a subset~$L \subseteq \Sigma^* \times \N$ over a finite alphabet~$\Sigma$.
Let~$f \colon \N \to \N$ be a computable function.
A problem $L$ is \emph{fixed-parameter tractable (in $\mathrm{FPT}$)} with respect to~$k$ if~$(I, k) \in L$ is decidable in time~$f(k) \cdot |I|^{O(1)}$ and $L$ is in XP if~$(I, k) \in L$ is decidable in time~$|I|^{f(k)}$.
There is a hierarchy of computational complexity classes for parameterized problems: $\mathrm{FPT} \subseteq \Wone \subseteq \Wtwo \subseteq \cdots \subseteq \mathrm{XP}$.
To show that a parameterized problem~$L$ is (presumably) not in FPT one may use a \emph{parameterized reduction} from a~$\Wone$-hard problem to~$L$.
A parameterized reduction from a parameterized problem~$L$ to another parameterized problem~$L'$ is a function that acts as follows:
For functions~$f$ and~$g$, given an instance~$(I, k)$ of~$L$, it computes in $f(k) \cdot |I|^{O(1)}$ time an instance~$(I', k')$ of~$L'$ so that~$(I, k) \in L \iff (I', k') \in L'$ and~$k' \le g(k)$.

\section{Hardness for Path-width and Feedback Vertex Number}
\label{sec:fvn}
In this section we show that \textsc{Geodetic Set} is~$\Wone$-hard with respect to the feedback vertex number, the path-width and the solution size, combined.
To this end, we present a parameterized reduction from \textsc{Grid Tiling}, which is~$\Wone$-hard with respect to~$k$ \cite{Marx07}:
\problemdef{Grid Tiling}
{A collection~$\mathcal S$ of $k^2$ sets~$S^{i, j} \subseteq [m] \times [m]$, $i, j \in [k]$ (called tile sets), each of cardinality exactly~$n$.}
{Can one choose a tile~$(x^{i, j}, y^{i, j}) \in S^{i, j}$ for each~$i, j \in [k]$ such that~$x^{i, j} = x^{i, j'}$ with~$j' = (j+1) \bmod k$ and~$y^{i, j} = y^{i', j}$ with~$i' = (i+1) \bmod k$?}
This distinguishes our reduction from most parameterized reductions to show~$\Wone$-hardness, as one typically reduces from \textsc{Clique}, or its multicolored variant.
\textsc{Grid Tiling} though seemed to be a much better fit, since the values of the tiles can be expressed by lengths of paths.
This is the central idea for our reduction:
We place a connection gadget between each pair of adjacent tile sets.
Placing paths of fitting lengths, the connection gadget ensures that the vertices corresponding to the tiles agree with each other, that is, the appropriate coordinates of the two tiles are equal.

\begin{remark}
	Throughout this section we write~$i'$ and~$j'$ as shorthands for~$(i+1) \bmod k$ and~$(j+1) \bmod k$, respectively.
	Moreover, we assume that the grid size~$k$ is even.
\end{remark}

\paragraph{Construction.}
Let~$I = (\mathcal S, k, m, n)$ be an instance of \textsc{Grid Tiling}.
We construct an instance of \textsc{Geodetic Set}~$I' = (G, k')$ as follows:
First, we set $k' = k^2 + 4$.
We add the \emph{global vertices}~$\Xi = \{\alpha, \beta, \gamma, \delta\}$ and~$\Xi' = \{\alpha', \beta', \gamma', \delta'\}$, and add four edges~$\alpha \alpha', \beta\beta',\allowbreak\gamma\gamma'$ and~$\delta\delta'$.
Next, for each~$i, j \in [k]$ we introduce \emph{tile vertices}~$S^{i, j} = \{s^{i,j}_1, \dots, s^{i, j}_n\}$.
For a tile vertex~$v$ we denote by~$(x_v, y_v)$ the corresponding tile.
Moreover, for each~$i, j \in [k]$ we introduce two copies of the horizontal and two copies of the vertical connection gadget.

The construction of a \emph{horizontal connection gadget next to tile set~$S^{i, j}$} is as follows.
Let~$S = S^{i, j}$ and let~$S' = S^{i, j'}$ be the vertices of the two horizontally adjacent tile sets.
We introduce the vertices~$a$ and~$b$ called \emph{hidden vertices} and the vertices~$a^*$ and~$b^*$ called \emph{exposed vertices}.
Next, for every tile vertex~$s \in S$ with its corresponding tile~$(x_s, y_s)$, we add
a path of length~$16m + 2x_s + 1$ from~$s$ to~$a$,
and a path of length~$16m - 2x_s + 1$ from~$s$ to~$b$.
For every tile vertex~$s' \in S'$ with its corresponding tile~$(x_{s'}, y_{s'})$, we add
a path of length~$16m - 2x_{s'} + 1$ from~$s'$ to~$a$,
and a path of length~$16m + 2x_{s'} + 1$ from~$s'$ to~$b$.
We call these paths \emph{tile paths towards~$S$}, respectively~$S'$.
We call the neighbors of~$a$, respectively~$b$, \emph{connector vertices towards~$S$}, respectively~$S'$.
The exposed vertices~$a^*$, respectively~$b^*$ are adjacent to all neighbors of~$a$, respectively~$b$.
Moreover, each of $a^*$ and $b^*$ has one additional neighbor:
If~$j$ is even, then~$\alpha$ is a neighbor of~$a^*$ and~$\beta$ is a neighbor of~$b^*$.
If~$j$ is odd, then~$\beta$ is a neighbor of~$a^*$ and~$\alpha$ is a neighbor of~$b^*$.
See \cref{fig:conn-gadget} (left) for an illustration of a horizontal connection gadget next to~$S^{i, j}$ for even~$j$.
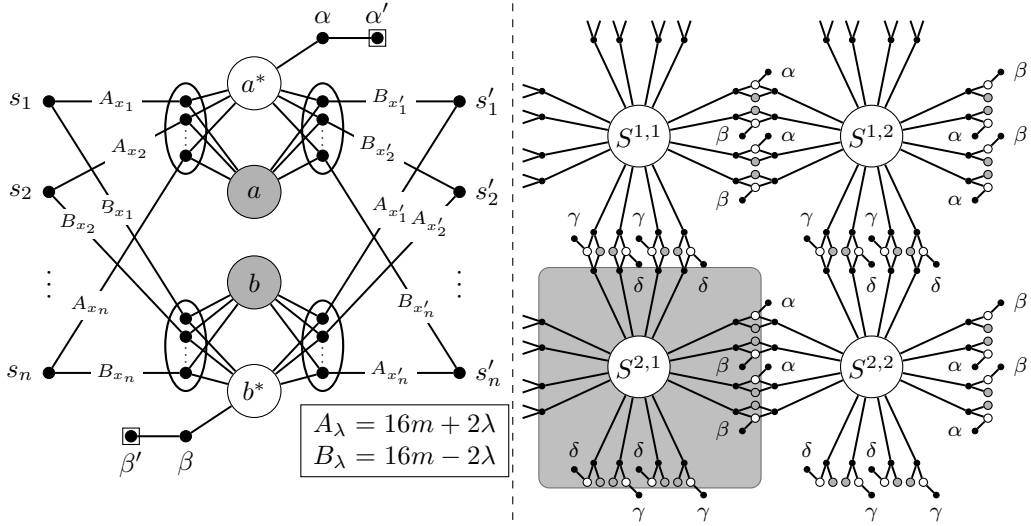
\begin{figure}[t]
	\centering
	\begin{tikzpicture}
		\begin{scope}[xscale=1.8,yscale=1.2]
			\tikzstyle{path} = [color=black,opacity=0.25,line cap=round, line join=round, line width=3.4mm]
			\draw[x radius=0.15cm,y radius=0.5cm,thick] (-0.5,1.7) circle;
			\draw[x radius=0.15cm,y radius=0.5cm,thick] (0.5,1.7) circle;
			\draw[x radius=0.15cm,y radius=0.5cm,thick] (-0.5,-0.7) circle;
			\draw[x radius=0.15cm,y radius=0.5cm,thick] (0.5,-0.7) circle;

			\node (s1) at (-1.5, 2)	[vertex, label={180:$s_1$}] {};
			\node (s2) at (-1.5, 1)	[vertex, label={180:$s_2$}] {};
			\node (s3) at (-1.5, 0)	[rotate=90] {$\dots$};
			\node (sn) at (-1.5, -1)[vertex, label={180:$s_n$}] {};

			\node (sp1) at (1.5, 2)	[vertex, label={0:$s'_1$}] {};
			\node (sp2) at (1.5, 1)	[vertex, label={0:$s'_2$}] {};
			\node (sp3) at (1.5, 0)	[rotate=90] {$\dots$};
			\node (spn) at (1.5, -1) [vertex, label={0:$s'_{n}$}] {};

			\node (a)	at (0, 1.0) [hollow,minimum width=20pt,fill=black!30] {$a$};
			\node (ap)	at (0, 2.2) [hollow,minimum width=20pt] {$a^*$};
			\node (app)	at (0.5, 2.7) [vertex, label={90:$\alpha$}] {};
			\node (alphap)	at (0.9, 2.7) [vertex, label={90:$\alpha'$}] {};
			\node (alphap-marked) at (alphap) [marked] {};
			\draw[edge] (ap) -- (app);
			\draw[edge] (app) -- (alphap);

			\node (b)	at (0, -0.0) [hollow,minimum width=20pt,fill=black!30] {$b$};
			\node (bp)	at (0, -1.2) [hollow,minimum width=20pt] {$b^*$};
			\node (bpp)	at (-0.5, -1.7) [vertex, label={270:$\beta$}] {};
			\node (betap)	at (-0.9, -1.7) [vertex, label={270:$\beta'$}] {};
			\node (betap-marked) at (betap) [marked] {};
			\draw[edge] (bp) -- (bpp);
			\draw[edge] (betap) -- (bpp);

			\node[draw,text width=2.5cm,align=center] at (1.1, -1.75) {
					$A_\lambda = 16m + 2\lambda$

					$B_\lambda = 16m - 2\lambda$
				};

			\node[vertex]	(ta1) at (-0.5, 2.0) {};
			\node[vertex]	(ta2) at (-0.5, 1.8) {};
			\node[rotate=90](ta3) at (-0.5, 1.6) {\tiny $\cdot\!\cdot\!\cdot$};
			\node[vertex]	(tan) at (-0.5, 1.4) {};

			\node[vertex]	(tap1) at (0.5, 2.0) {};
			\node[vertex]	(tap2) at (0.5, 1.8) {};
			\node[rotate=90](tap3) at (0.5, 1.6) {\tiny $\cdot\!\cdot\!\cdot$};
			\node[vertex]	(tapn) at (0.5, 1.4) {};

			\node[vertex]	(tb1) at (-0.5, -0.4) {};
			\node[vertex]	(tb2) at (-0.5, -0.6) {};
			\node[rotate=90](tb3) at (-0.5, -0.8) {\tiny $\cdot\!\cdot\!\cdot$};
			\node[vertex]	(tbn) at (-0.5, -1.0) {};

			\node[vertex]	(tbp1) at (0.5, -0.4) {};
			\node[vertex]	(tbp2) at (0.5, -0.6) {};
			\node[rotate=90](tbp3) at (0.5, -0.8) {\tiny $\cdot\!\cdot\!\cdot$};
			\node[vertex]	(tbpn) at (0.5, -1.0) {};

			\foreach \x / \name / \posa / \posb in {1/1/0.5/0.5, 2/2/0.6/0.2, n/n/0.3/0.5}{
				
				\draw[edge] (s\x)  to node[pos=\posa, fill=white, inner sep=2pt] {\scriptsize $A_{x_\name}$} (ta\x);
				\draw[edge] (sp\x) to node[pos=\posa, fill=white, inner sep=2pt] {\scriptsize $B_{x'_\name}$} (tap\x);
				\draw[edge] (s\x)  to node[pos=\posb, fill=white, inner sep=2pt] {\scriptsize $B_{x_\name}$} (tb\x);
				\draw[edge] (sp\x) to node[pos=\posb, fill=white, inner sep=2pt] {\scriptsize $A_{x'_\name}$} (tbp\x);

				\draw[edge] (ta\x)  to (a);
				\draw[edge] (tap\x) to (a);
				\draw[edge] (tb\x)  to (b);
				\draw[edge] (tbp\x) to (b);
				\draw[edge] (ta\x)  to (ap);
				\draw[edge] (tap\x) to (ap);
				\draw[edge] (tb\x)  to (bp);
				\draw[edge] (tbp\x) to (bp);
			}

		\end{scope}
		\draw[dashed] (3.40, 3.7) -- (3.40, -3.15);
		\begin{scope}[xshift=2.0cm,yshift=5.0cm,scale=1.7]

			\begin{pgfonlayer}{bg}
				\draw[rounded corners,fill=gray,opacity=0.5] (1.025, -2.825) rectangle (2.745, -4.545) {};
			\end{pgfonlayer}
			\foreach \x in {1,2}{
				\foreach \y in {1,2}{
					\begin{scope}[xshift=1.8*\y cm,yshift=-1.8*\x cm]
						\node[vertex,fill=white,minimum width=.75cm] (s) at (0,0) {$S^{\x, \y}$};
						\begin{pgfonlayer}{bg}
						\node[]	(la1p) at (-.9, .40) {};
						\node[]	(la1)  at (-.9, .30) {};
						\node[]	(lb1)  at (-.9, .20) {};
						\node[]	(lb1p) at (-.9, .10) {};
						\node[smallvertex]	(lta1) at (-.75, .35) {};
						\node[smallvertex]	(ltb1) at (-.75, .15) {};

						\node[]	(la2p) at (-.9, -.10) {};
						\node[]	(la2)  at (-.9, -.20) {};
						\node[]	(lb2)  at (-.9, -.30) {};
						\node[]	(lb2p) at (-.9, -.40) {};
						\node[smallvertex]	(lta2) at (-.75,-.15) {};
						\node[smallvertex]	(ltb2) at (-.75,-.35) {};

						\node[]	(uc1p) at (-.40, 0.9)  {};
						\node[]	(uc1)  at (-.30, 0.9)  {};
						\node[]	(ud1)  at (-.20, 0.9)  {};
						\node[]	(ud1p) at (-.10, 0.9)  {};
						\node[smallvertex]	(utc1) at (-.35, 0.75) {};
						\node[smallvertex]	(utd1) at (-.15, 0.75) {};

						\node[]	(uc2p) at (.40, 0.9)  {};
						\node[]	(uc2)  at (.30, 0.9)  {};
						\node[]	(ud2)  at (.20, 0.9)  {};
						\node[]	(ud2p) at (.10, 0.9)  {};
						\node[smallvertex]	(utc2) at (.35, 0.75) {};
						\node[smallvertex]	(utd2) at (.15, 0.75) {};

						\draw[edge] (s) -- (lta1);
						\draw[edge] (s) -- (ltb1);
						\draw[edge] (s) -- (utc1);
						\draw[edge] (s) -- (utd1);
						\draw[edge] (s) -- (lta2);
						\draw[edge] (s) -- (ltb2);
						\draw[edge] (s) -- (utc2);
						\draw[edge] (s) -- (utd2);
						\draw[edge] (lta1) -- (la1p.center);
						\draw[edge] (lta1) -- (la1.center);
						\draw[edge] (ltb1) -- (lb1.center);
						\draw[edge] (ltb1) -- (lb1p.center);
						\draw[edge] (utc1) -- (uc1p.center);
						\draw[edge] (utc1) -- (uc1.center);
						\draw[edge] (utd1) -- (ud1.center);
						\draw[edge] (utd1) -- (ud1p.center);
						\draw[edge] (lta2) -- (la2p.center);
						\draw[edge] (lta2) -- (la2.center);
						\draw[edge] (ltb2) -- (lb2.center);
						\draw[edge] (ltb2) -- (lb2p.center);
						\draw[edge] (utc2) -- (uc2p.center);
						\draw[edge] (utc2) -- (uc2.center);
						\draw[edge] (utd2) -- (ud2.center);
						\draw[edge] (utd2) -- (ud2p.center);
						\end{pgfonlayer}

						\node[tinycircle]	(a1p) at (0.9, .40) {};
						\node[tinycircle,gray]	(a1)  at (0.9, .30) {};
						\node[tinycircle,gray]	(b1)  at (0.9, .20) {};
						\node[tinycircle]	(b1p) at (0.9, .10) {};
						\node[smallvertex]	(ta1)  at (0.75, .35) {};
						\node[smallvertex]	(tb1)  at (0.75, .15) {};

						\node[tinycircle]	(a2p) at (0.9, -.10) {};
						\node[tinycircle,gray]	(a2)  at (0.9, -.20) {};
						\node[tinycircle,gray]	(b2)  at (0.9, -.30) {};
						\node[tinycircle]	(b2p) at (0.9, -.40) {};
						\node[smallvertex]	(ta2) at (0.75,-.15) {};
						\node[smallvertex]	(tb2) at (0.75,-.35) {};

						\node[tinycircle]	(c1p) at (-.40, -.9) {};
						\node[tinycircle,gray]	(c1)  at (-.30, -.9) {};
						\node[tinycircle,gray]	(d1)  at (-.20, -.9) {};
						\node[tinycircle]	(d1p) at (-.10, -.9) {};
						\node[smallvertex]	(tc1) at (-.35,-.75) {};
						\node[smallvertex]	(td1) at (-.15,-.75) {};

						\node[tinycircle]	(c2p) at (.10, -.9) {};
						\node[tinycircle,gray]	(c2)  at (.20, -.9) {};
						\node[tinycircle,gray]	(d2)  at (.30, -.9) {};
						\node[tinycircle]	(d2p) at (.40, -.9) {};
						\node[smallvertex]	(tc2) at (.15,-.75) {};
						\node[smallvertex]	(td2) at (.35,-.75) {};

						\draw[edge] (s) -- (ta1);
						\draw[edge] (s) -- (tb1);
						\draw[edge] (s) -- (tc1);
						\draw[edge] (s) -- (td1);
						\draw[edge] (s) -- (ta2);
						\draw[edge] (s) -- (tb2);
						\draw[edge] (s) -- (tc2);
						\draw[edge] (s) -- (td2);
						\draw[edge] (ta1) -- (a1p);
						\draw[edge] (ta1) -- (a1);
						\draw[edge] (tb1) -- (b1);
						\draw[edge] (tb1) -- (b1p);
						\draw[edge] (tc1) -- (c1p);
						\draw[edge] (tc1) -- (c1);
						\draw[edge] (td1) -- (d1);
						\draw[edge] (td1) -- (d1p);
						\draw[edge] (ta2) -- (a2p);
						\draw[edge] (ta2) -- (a2);
						\draw[edge] (tb2) -- (b2);
						\draw[edge] (tb2) -- (b2p);
						\draw[edge] (tc2) -- (c2p);
						\draw[edge] (tc2) -- (c2);
						\draw[edge] (td2) -- (d2);
						\draw[edge] (td2) -- (d2p);

						\node[smallvertex,label={0:\footnotesize\ifthenelse{\y=1 \OR \y=3}{$\alpha$}{$\beta$}}] (alpha1) at (1.0, 0.50) {};
						\node[smallvertex,label={0:\footnotesize\ifthenelse{\y=1 \OR \y=3}{$\alpha$}{$\beta$}}] (alpha2) at (1.0, -.00) {};
						\node[smallvertex,label={180:\footnotesize\ifthenelse{\y=2 \OR \y=4}{$\alpha$}{$\beta$}}] (beta1) at (0.8, 0.00) {};
						\node[smallvertex,label={180:\footnotesize\ifthenelse{\y=2 \OR \y=4}{$\alpha$}{$\beta$}}] (beta2) at (0.8, -.50) {};
						\node[smallvertex,label={90:\footnotesize\ifthenelse{\x=1 \OR \x=3}{$\gamma$}{$\delta$}}] (gamma1) at (-.50, -.8) {};
						\node[smallvertex,label={90:\footnotesize\ifthenelse{\x=1 \OR \x=3}{$\gamma$}{$\delta$}}] (gamma2) at (0.00, -.8) {};
						\node[smallvertex,label={270:\footnotesize\ifthenelse{\x=2 \OR \x=4}{$\gamma$}{$\delta$}}] (delta1) at (-.00, -1.) {};
						\node[smallvertex,label={270:\footnotesize\ifthenelse{\x=2 \OR \x=4}{$\gamma$}{$\delta$}}] (delta2) at (0.50, -1.) {};
						\draw[edge] (a1p) -- (alpha1);
						\draw[edge] (b1p) -- (beta1);
						\draw[edge] (c1p) -- (gamma1);
						\draw[edge] (a2p) -- (alpha2);
						\draw[edge] (b2p) -- (beta2);
						\draw[edge] (c2p) -- (gamma2);
						\draw[edge] (d1p) -- (delta1);
						\draw[edge] (d2p) -- (delta2);
					\end{scope}
				}
			}
		\end{scope}
	\end{tikzpicture}
	\caption{
		\emph{Left:}
		One copy of a horizontal connection gadget next to~$S^{i, j} = \{s_1, \dots, s_n\}$ where~$j$ is even, connecting the tile sets~$S^{i, j}$ and~$S^{i, j'}$.
		Edges with label~$\ell$ in the figure represent paths of length~$\ell$.
		The ellipses mark the connector vertices towards~$S^{i, j}$ and~$S^{i, j'}$.
		\emph{Right:}
		An exemplary reduction from an instance of \textsc{Grid Tiling}, where~$k = 2$.
		Between every pair of horizontally, resp.\ vertically adjacent tile sets (big circles) there are two copies of horizontal, resp.\ vertical connection gadgets.
		Note that~$\alpha, \beta, \gamma, \delta \in \Xi$ are global; every vertex labeled such is the same vertex.
		The gray square marks the vertices of~$Q^{2, 1}$ (note that~$\beta, \delta \notin Q^{3, 2}$).
		Note that this illustration wraps around its boundaries.
	}
	\label{fig:conn-gadget}
\end{figure}

The construction of a \emph{vertical connection gadget next to tile set~$S^{i, j}$} is identical to the construction of a horizontal gadget, except for the following differences:
\begin{itemize}
	\item the gadget connects tile sets~$S = S^{i, j}$ and~$S' = S^{i', j}$;
	\item the lengths of the tile paths depend on the~$y$-coordinates; and
	\item if~$i$ is even, then~$\gamma$ is a neighbor of~$a^*$ and~$\delta$ is a neighbor of~$b^*$, and if~$i$ is odd, then~$\delta$ is a neighbor of~$a^*$ and~$\gamma$ is a neighbor of~$b^*$.
\end{itemize}

This concludes the construction.
See \cref{fig:conn-gadget} (right) for an overview.

Let~$J$ be the set of all hidden vertices and let~$J^*$ be the set of all exposed vertices.
We now show that this construction has the desired properties for showing~$\Wone$-hardness with respect to solution size, feedback vertex number and path-width, combined.
\begin{observation}
	\label{obs:w1-bounded-parameters}
	The constructed graph~$G$ has~$\pw(G) \le 16k^2 +2$ and~$\fvn(G) \le 16k^2$.
\end{observation}
\begin{proof}
	The graph~$G' = G - (J \cup J^*)$ consists of paths of length one and subdivisions of stars.
	Clearly, $\fvn(G') = 0$, and since removing the center vertex of a subdivision of a star yields disjoint paths, $\pw(G') = 2$.
	Adding a vertex to a graph increases each of the two parameters by at most one.
	Now, as~$|J \cup J^*| = 16k^2$, the claim follows.
\end{proof}

\paragraph{Correctness.}
Let us first point out that the central challenge is to cover all hidden vertices~$J$, as every other vertex is covered by the four degree-one vertices in~$\Xi'$.

\begin{observation}
	\label{obs:allbuthiding}
	$I[\Xi'] = V(G) \setminus J$.
\end{observation}
\begin{proof}
	For~$i, j \in [k]$, for~$s \in S^{i, j}$ let~$(x_s, y_s) \in [m] \times [m]$ be the values of the corresponding tile.
	We show first that all vertices in horizontal connection gadgets are covered.
	Suppose first that~$j$ is even.
	For every~$s \in S^{i, j'}$, there are~$16$ shortest~$\alpha'$--$\beta'$-paths of length~$3 + 16m + 2x_s + 16m + 2x_s + 3 = 32m+6$, each of which is also a shortest~$s$-visiting path.
	Take one each of the two horizontal connection gadgets next to~$S^{i, j}$ and~$S^{i, j'}$, and denote by~$a^*, b^*$, respectively~${a'}^*, {b'}^*$ the exposed vertices to the corresponding connection gadgets.
	Then we have the following shortest~$\alpha'$--$\beta'$-paths via~$s$:
	(1)~one path via~$a^*$, $s$, and~$b^*$,
	(2)~one path via~${b'}^*$, $s$, and~${a'}^*$,
	(3)~one path via~$a^*$, $s$, and~${a'}^*$, and
	(4)~one path via~${b'}^*$, $s$, and~$b^*$.
	The paths described in~(3) and~(4) also use vertices in the horizontal connection gadget next to~$S^{i, j'}$.
	Note that since~$j'$ is odd, the exposed vertex~${a'}^*$ is adjacent to~$\beta$.

	The case that~$j$ is odd behaves analogously.
	Note that~$\alpha$ now is adjacent to the exposed vertex~$b^*$ while~$\beta$ is connected to~$a^*$.
	This gives us that the shortest~$\alpha'$--$\beta'$-paths cover all tile vertices as well as all vertices in horizontal connection gadgets, except for the hidden vertices.
	
	Symmetry gives us that the shortest~$\gamma'$--$\delta'$-paths cover all tile vertices as well as all vertices in vertical connection gadgets, except for the hidden vertices;
	thus~$V(G) \setminus J \subseteq I[\Xi']$.

	It remains to be shown that~$J \cap I[\Xi'] = \emptyset$.
	Note that the neighborhood of any hidden vertex is a subset of the neighborhood of the corresponding exposed vertex.
	Since each vertex in~$\Xi$ is connected to exactly one vertex in~$\Xi'$ and to exposed vertices, $I[\Xi']$ cannot contain any hidden vertex.
\end{proof}
Then the forward direction becomes straightforward: Our geodetic set~$V'$ consists of~$\Xi'$ and, for every tile in the solution of instance~$I$, the corresponding tile vertex.
It is easy to see that for every (copy of a) connection gadget, there are two shortest paths between the chosen tile vertices of any two adjacent tiles, each covering one of the two hidden vertices in the connection gadget.
Compare with \cref{fig:conn-gadget} (hidden vertices are gray).

The backward direction is more involved.
We show in two steps that every solution of our constructed instance consists of~$\Xi'$ and exactly one tile vertex of each tile set.
For this we make use of two properties of our construction.
First, if two vertices are sufficiently far apart, then there is a shortest path via some global vertex that connects them.
\begin{lemma}
	\label{lem:short-global-path}
	For any two vertices~$u, v \in V(G)$ there is a~$u$--$v$-path of length at most~$36m + 6$ that visits some global vertex.
\end{lemma}
\begin{proof}
	We define $\xi_u \in \Xi$ as follows.
	If~$u \in J \cup J^* \cup \Xi \cup \Xi'$, then let~$\xi_u \in \Xi$ be an arbitrary global vertex such that~$d(u, \xi_u) \le d(u, \zeta)$ for all~$\zeta \in \Xi$.
	Suppose that~$u$ is in a (horizontal or vertical) connection gadget.
	Then~$u$ lies on a path between a tile vertex~$u' \in S^{i, j}$, and a connector vertex~$u''$ towards~$S^{i, j}$, where $i, j \in [k]$.
	Let~$\xi_u \in \Xi$ be a global vertex such that~$d(u'', \xi_u) \le d(u'', \zeta)$, for~$\zeta \in \Xi$.
	We define $\xi_v$ analogously.
	If~$\xi_u = \xi_v$, then~$d(u, v) \le d(u, \xi_u) + d(\xi_u, v) \le 16m + 2\lambda + 2 + 2 + 2\lambda' + 16m \le 36m + 6$, where~$\lambda, \lambda' \in [m]$ are either~$x$- or~$y$-values of some tile.

	So suppose that $\xi_u \ne \xi_v$.
	%By definition we have 
	%\begin{align*}
	%	d(u, \xi_u) + d(\xi_u, v) + d(u, \xi_v) + d(\xi_v, v)
	%	d(u, \xi_u, v) + d(u, \xi_v, v) &= d(u, \xi_v) + d(\xi_v, v) + d(v, \xi_u) + d(\xi_u, u)\\
	%					&= d(\xi_u, u, \xi_v) + d(\xi_v, v, \xi_u).
	%\end{align*}
	We will prove that
	\begin{equation*}
		\label{eq:twotimesdiam}
		\begin{aligned}
		d(u, \xi_u) + d(\xi_u, v) + d(u, \xi_v) + d(\xi_v, v) =& \ d(\xi_u, u) + d(u, \xi_v) + d(\xi_v, v) + d(v, \xi_u)\\
		\le& \ 2(36m + 6),
		\end{aligned}
	\end{equation*}
	which yields the statement above as $d(u, v) \le \min \{ d(u, \xi_u) + d(\xi_u, v), d(u, \xi_v) + d(\xi_v, v) \}$.
	In particular, we show that $d(\xi_u, u) + d(u, \xi_v) \le 36m + 6$.
	If~$u \notin J$, then $d(\xi_u, u) + d(u, \xi_v) \le d(\xi_u, u') + d(u', \xi_v)$ for some tile vertex $u'$.
	Thus we obtain 
	\[
	%\begin{align*}
		d(\xi_u, u) + d(u, \xi_v) \le d(\xi_u, u') + d(u', \xi_v) = 2 + 16m + 2\lambda + 16m + 2\lambda' + 2 \le 36m + 4,
	%\end{align*}
	\]
	where~$\lambda, \lambda' \in [m]$ are either~$x$- or~$y$-values of some tile.
	If~$u \in J$, then we have
	\[
	%\begin{align*}
		d(\xi_u, u) + d(u, \xi_v) = 3 + 1 + 16m + 2\lambda + 16m + 2\lambda' + 2 \le 36m + 6.
	%\end{align*}
	\]
	Analogously, $d(\xi_v, v) + d(v, \xi_u) \le 36m + 6$, concluding the proof.
\end{proof}

With \cref{lem:short-global-path} at hand, it is easy to derive from \cref{fig:conn-gadget} (left) the following observation, which is also the reason why the vertices in~$J$ are called \emph{hidden}.

\begin{observation}
	\label{obs:hiding-vertices}
	Let~$u, v \in V(G) \setminus (\Xi \cup \Xi')$.
	If a shortest~$u$--$v$-path visits a global vertex, then none of its inner vertices is a hidden vertex.
\end{observation}

We introduce some additional notation.
The \emph{square} $Q^{i, j}$ of tile set~$S^{i, j}$ is the vertex set consisting of the tile vertices~$S^{i, j}$, the paths between tile vertices and connector vertices towards~$S^{i, j}$, and all hidden vertices and exposed vertices that are in the connection gadgets next to~$S^{i, j}$.
See \cref{fig:conn-gadget} (right) for an illustration of a square.
Note that the squares are pairwise disjoint.
We say that two squares are adjacent if they contain vertices of the same connection gadget.
The \emph{adjacency}~$\Adj(Q^{i, j})$ of a square~$Q^{i, j}$ is the union of squares adjacent to~$Q^{i, j}$.
The \emph{closed adjacency} of a square~$Q^{i, j}$ is the vertex set~$\Adj[Q^{i, j}] = \Adj(Q^{i, j}) \cup Q^{i, j}$.

We show that any solution of $(G, k')$ contains exactly one vertex per square.

\begin{lemma}
	\label{lem:w1-one-per-q}
	A geodetic set~$V' \subseteq V(G)$ of size at most~$k'$ consists of the four vertices in~$\Xi'$, and exactly one vertex in each square~$Q^{i, j}$, for each $i, j \in [k]$.
\end{lemma}
\begin{proof}
	Recall that~$k' = k^2 + 4$.
	The four vertices in~$\Xi'$ are the only vertices of degree one and are part of every geodetic set.
	Further we may assume that~$V' \cap \Xi = \emptyset$ as~$I[V'] = I[V' \setminus \Xi]$.
	So~$V'$ consists of the four vertices in~$\Xi'$ and a set of at most~$k^2$ vertices within the squares, denoted by~$W$.

	For contradiction, assume that there are $q > 0$ squares $Q_1, \dots, Q_q$ such that $Q_p \cap W = \emptyset$ for $p \in [q]$.
	We call these squares \emph{empty}, and all other squares \emph{non-empty}.
	We claim that there is an empty square~$Q_p$ such that~$|\Adj(Q_p) \cap W| \le 8$.
	Let~$W' \subseteq W$ be an arbitrary subset consisting of \emph{exactly one vertex of~$W$\! per non-empty square}.
	So~$|W'| = k^2 - q$ and $|W \setminus W'| \le q$.
	Clearly, for each~$p \in [q]$, we have $|\Adj(Q_p) \cap W'| \le 4$, thus $\sum_{p=1}^q |\Adj(Q_p) \cap W'| \le 4q$.
	Since $\sum_{p = 1}^q |\Adj(Q_p) \cap \{ v \}| \le 4$ for any vertex $v \in V(G)$, we also have
	\[
		\sum_{p = 1}^q | \Adj(Q_p) \cap (W \setminus W')| = \sum_{p = 1}^q \sum_{v \in W \setminus W'} | \Adj(Q_p) \cap \{ v \}| \le 4q.
	\]
	Consequently,
	\[
		\sum_{p=1}^q |\Adj(Q_p) \cap W|
		= \sum_{p=1}^q |\Adj(Q_p) \cap W'| + \sum_{p=1}^q |\Adj(Q_p) \cap (W\setminus W')| \le 4q + 4q = 8q.
	\]
	It follows that there exists an empty square~$Q$ for which~$|\Adj(Q) \cap W| \le 8$.

	Let~$J_Q = J \cap N[Q]$ be the sixteen hidden vertices that are either in~$Q$ or adjacent to vertices of~$Q$.
	The next two claims are consequences of \cref{lem:short-global-path,obs:hiding-vertices}:
	\begin{enumerate}[(1)]
		\item no shortest path between a vertex outside of~$Q$ and a vertex outside of~$\Adj[Q]$ can visit any vertex in~$J_Q$, and
		\item $W$ covers at most~$|\Adj(Q) \cap W| \le 8$ vertices of~$J_Q$.
	\end{enumerate}

	For (1), let~$u \in V(G) \setminus Q$, let~$v \in J_Q$ and let~$w \in V(G) \setminus \Adj[Q]$ (possibly~$v = w$).
	Observe first that any shortest path that visits~$v$ and whose endpoints are not in the connection gadget containing~$v$ visits tile vertices of both incident tiles.
	Also, the shortest path cannot visit any of the global vertices~$\Xi$ as they provide a short cut around~$v$.
	Then it is easy to see that any shortest $u$--$w$-path visiting~$v$ must at some point go through vertices of some square~$Q' \subseteq \Adj(Q)$, then visit~$v$, enter~$Q$, and then leave~$Q$ again before reaching~$w$.
	Within~$Q'$, such a path covers a distance of at least~$16m - 2\lambda + 16m - 2\lambda' \ge 28m$ for~$\lambda, \lambda' \in [m]$.
	Analogously, it covers a distance of at least~$28m$ within~$Q$ as well.
	Thus its length is at least~$56m$, and by \cref{lem:short-global-path} such a path is longer than~$\diam(G)$, contradicting the existence of a shortest~$u$-$w$-path that visits a vertex in~$J_Q$.

	For (2), suppose that there exist $u, u' \in \Adj(Q)$ such that there is a shortest $u$--$u'$-path~$P$ that visits~$v \in J_Q$.
	It is easy to see that $P$ must go through $v' \ne v \in J_Q$.
	Assume without loss of generality that $v$ appears before $v'$ in $P$.
	In order for $P$ to be a shortest path, it must hold that
	\[
		d(u, v) + d(v, v') + d(v', u') \le d(u, u') \le 36m + 6,
	\]
	due to \cref{lem:short-global-path}.
	Since~$d(v, v') = 2 + 16m - 2\lambda + 16m - 2\lambda' + 2$ for some $\lambda, \lambda' \in [m]$, we have~$d(v, v') \ge 28m+4$; thus we can assume that~$d(u, v) + d(u', v') \le 8m + 2$.
	Then, by construction, $u$ (respectively~$u'$) lies on a path between a tile vertex and~$v$ (respectively~$v'$).
	By (1), only the vertices in~$W \cap \Adj(Q)$ can cover the vertices in~$J_Q$.
	Hence, for every vertex~$u \in W \cap \Adj(Q)$ there is exactly one vertex~$v \in J_Q$ that is going to be in~$I[W]$, and the claimed inequality holds.

	Since~$|J_Q| = 16$, the set $V'$ is not geodetic; so there cannot be an empty square in~$G$.
	There are~$k^2$ squares and~$|W| = |V' \setminus \Xi'| \le k^2$.
	So~$|V' \cap Q^{i, j}| = 1$ for each~$i, j \in [k]$.
\end{proof}

Using \cref{lem:w1-one-per-q}, we show that every solution vertex in a square must be a tile vertex.

\begin{lemma}
	\label{lem:w1-only-tile-vertices}
	A geodetic set~$V' \subseteq V(G)$ of size at most $k'$ consists of the four vertices in~$\Xi'$ and exactly one vertex of~$S^{i, j}$, for each $i, j \in [k]$.
\end{lemma}
\begin{proof}
	For~$i, j \in [k]$, let~$S = S^{i, j}$, $S' = S^{i, j'}$,~$Q = Q^{i, j}$, and~$Q' = Q^{i, j'}$.
	Without loss of generality, assume that $j$ is even (see \cref{fig:conn-gadget} for an illustration).
	Let~$X_1$ and~$X_2$ be the two copies of the horizontal connection gadget next to tile~$S$, let~$a_1, b_1 \in V(X_1)$ and~$a_2, b_2 \in V(X_2)$ be the hidden vertices, and let~$a^*_1,  b^*_1 \in V(X_1)$ and $a^*_2, b^*_2 \in V(X_2)$ be the exposed vertices.
	By \cref{lem:w1-one-per-q}, $V'$ contains exactly one vertex $u$ in~$Q$ and exactly one vertex~$v$ in~$Q'$.

	Consider a vertex~$w \in V(G) \setminus (Q \cup Q')$.
	Note that any shortest~$u$--$w$-path and any shortest~$v$--$w$-path going through one of~$a_{1}, a_{2}, b_{1}, b_{2}$ must use tile vertices in~$S$ and~$S'$.
	It is easy to verify that due to its length, such a path must visit some global vertex, thus it cannot visit any hidden vertex (\cref{obs:hiding-vertices}).
	It follows that $\{ a_1, a_2, b_1, b_2 \} \subseteq I[u, v]$.

	For the sake of contradiction, suppose that $u \notin S$.
	In particular, we assume without loss of generality that~$u \in V(X_1)$.
	Let~$u' \in S$ be the tile vertex such that~$u$ lies on the tile path between~$u'$ and $a_1$.
	Observe that~$d(u, a_1) < d(u, a_2)$.
	Hence, no shortest~$u$--$v$-path visits~$a_2$ if $d(v, a_1) \le d(v, a_2)$.
	It follows that $v$ lies on some tile path between~some tile vertex $v' \in S'$ and $a_2$.
	Since there are shortest~$u$--$v$-paths visiting $a_{1}$ and $a_{2}$, we have
	\begin{align*}
		d(u, v) &= (d(a_{1}, u') - d(u', u)) + d(a_{1}, v') + d(v, v') \text{ and } \\
		d(u, v) &= (d(a_2, v') - d(v, v')) + d(a_2, u') + d(u, u').
	\end{align*}
	By construction, $d(a_{1}, u') = d(a_{2}, u') = 16m + 2 x_{u'} + 1$ and $d(a_1, v') = d(a_2, v') = 16m - 2 x_{v'} + 1$.
	Thus, we obtain $d(u, u') = d(v, v')$ and $d(u, v) = 32m + 2 x_{u'} - 2 x_{v'} + 2$.
	Note that there is a $u$--$v$-path visiting $\alpha$ that is of length%
	\[
		\ell = (d(u', a_1^*) - d(u, u')) + 2 + (d(a_2^*, v') - d(v', v)).
	\]
	Since $d(a_1, u') = 16m + 2 x_{u'} + 1$ and $d(v', a_1) = 16m - 2 x_{v'} + 1$ (by construction), and since $\ell \ge d(u, v)$, we obtain~$d(u, u') = d(v, v') \le 1$.
	By the assumption that $u \notin S$, we have $d(u, u') > 0$.
	It follows that $d(u, u') = d(v, v') = 1$.
	Finally, observe that the shortest path from~$u$ to~$v$ that visits~$b_1$ is of length%
	\[
		\ell' = d(u, b_{1}) + d(b_1, v) = 32m - 2x_{u'} + 2x_{v'} + 4.
	\]
	Since $\ell' = d(u, v)$, we obtain $4x_{u'} - 4x_{v'} = 2$, so one of~$x_{u'}, x_{v'}$ cannot be integer---a contradiction.
\end{proof}

Now, given \cref{lem:w1-only-tile-vertices}, if there is a solution for our instance of \textsc{Geodetic Set}, then the tiles corresponding to the chosen tile vertices are a solution for our instance of \textsc{Grid Tiling}.
The main theorem of the section follows:

\begin{theorem}
	\label{thm:w1-fvn}
	\emph{\textsc{Geodetic Set}} is~$\mathrm{W}[1]$-hard with respect to the feedback vertex number, the path-width, and the solution size, combined.
\end{theorem}
\begin{proof}
	Given an instance~$(\mathcal S, k, m, n)$ of \textsc{Grid Tiling}, we construct an instance~$(G, k')$ of \textsc{Geodetic Set} as shown above.
	We now prove that~$(\mathcal S, k, m, n)$ is a yes instance if and only if~$(G, k')$ is a yes-instance.

	For the \emph{if} direction, let~$\mathcal S' = \{(x^{i, j}, y^{i, j}) \in S^{i, j} \mid i, j \in [k] \}$ be a solution for the instance $(\mathcal S, k, m, n)$.
	Then we construct a geodetic set~$V'$ by adding the vertices in~$\Xi'$ and, for every~$i, j \in [k]$, the tile vertex~$s^{i, j} \in S^{i, j}$, corresponding to~$(x^{i, j}, y^{i, j})$.
	Clearly, $|V'| = k' = k^2 + 4$.
	By \cref{obs:allbuthiding}, all vertices in~$V(G) \setminus J$ are covered.
	For~$i, j \in [k]$, let~$a$ and~$b$ be the hidden vertices of one of the copies of the horizontal connection gadget next to~$S^{i, j}$.
	Since~$x^{i,j} = x^{i, j'}$, the shortest~$a$-visiting~$s^{i, j}$--$s^{i, j'}$-paths have length
	\[ d(s^{i, j}, a, s^{i, j'})= 16m + 2x^{i, j} + 1 + 1 + 16m - 2x^{i, j'} = 32m + 2,\]
	and the shortest~$b$-visiting~$s^{i, j}$--$s^{i, j'}$-paths have length
	\[ d(s^{i, j}, b, s^{i, j'})= 16m - 2x^{i, j} + 1 + 1 + 16m + 2x^{i, j'} = 32m + 2.\]
	It is easy to see that there are no shorter~$s^{i, j}$--$s^{i, j'}$-paths.
	So the hidden vertices of the two copies of the horizontal connection gadget are in~$I[V']$.
	Analogously, since~$y^{i, j} = y^{i', j}$, there exist shortest~$v^{i, j}$--$v^{i', j}$-paths that visit
	the hidden vertices of the two copies of the vertical connection gadgets next to~$S^{i, j}$.
	Thus~$V'$ is geodetic and of cardinality~$k'$.

	For the \emph{only if} direction, let~$V'$ be a solution for~$(G, k')$.
	By \cref{lem:w1-only-tile-vertices}, $V'$ consists only of the vertices in~$\Xi'$ and one tile vertex~$s^{i, j}$ for each $i, j \in [k]$.
	Let~$(x^{i, j}, y^{i, j})$ be the corresponding pair.
	Note that the hidden vertices within the copies of~$X^{i, j}$ con only be covered by shortest~$s^{i, j}$--$s^{i, j'}$-paths.
	But in order for these paths to be of equal length, it must hold that~$x^{i, j} = x^{i, j'}$.
	Analogously, in order to cover the hidden vertices within the copies of~$Y^{i, j}$, we must have~$y^{i, j} = y^{i', j}$.
	So choosing the pair~$(x^{i, j}, y^{i, j})$ for each~$i, j \in [k]$ yields a solution for the instance~$(\mathcal S, k, m, n)$ of \textsc{Grid Tiling}.

	Since \textsc{Grid Tiling} is~$\Wone$-hard with respect to~$k$, it follows from the reduction and \cref{obs:w1-bounded-parameters} that \textsc{Geodetic Set} is~$\Wone$-hard with respect to~$k' + \fvn(G) + \pw(G)$.
\end{proof}

\section{Fixed-Parameter Tractability for Feedback Edge Number}
\label{sec:fen}
We now show that \textsc{Geodetic Set} is fixed-parameter tractable for feedback edge number.
In fact, we present a fixed-parameter algorithm for the following, more general variant:

\problemdef{Extended Geodetic Set}
{A graph~$G$, a vertex set $T \subseteq V(G)$, and an integer~$k$.}
{Does~$G$ have a geodetic set $S \supseteq T$ of cardinality at most~$k$?}

The algorithm works in three steps:
We first apply some polynomial-time data reduction rules.
The graph may be arbitrarily large even after they are applied exhaustively.
However, together with some branching steps, they lead to an instance in which a part of the solution vertices are fixed and can be extended to a minimum geodetic set by adding vertices on paths of degree-two vertices.
We determine these vertices using an ILP formulation with $O(\fen(G)^2)$ variables, showing that \textsc{(Extended) Geodetic Set} is fixed-parameter tractable for feedback edge number.

Although feedback edge number is considered one of the largest structural graph parameters, our algorithm is still technically involved and it has an impractical running time.
This hints at the difficulty of designing efficient algorithms for \textsc{Geodetic Set}.
We also remark that some of the techniques presented may be of independent interest.
For example, the presented approach may also be useful to show fixed-parameter tractability of the closely related \textsc{Metric Dimension} problem\footnote{Given a graph, \textsc{Metric Dimension} asks for a set~$S$ of at most $k$ vertices such that for any pair of vertices~$u$ and~$v$, there is a vertex in~$S$ which has distinct distances to~$u$ and~$v$.} for feedback edge number, which was posed as an open problem by Eppstein~\cite{Epp15} (so far, it is only known to be in XP for this parameter~\cite{ELW15}).

This section is divided into three parts.
In \cref{sec:preprocessing}, we provide some polynomial-time data reduction rules, which allow us to bound the number of vertices with degree at least three.
In \cref{sec:guessing}, we guess parts of the solution.
Finally, in \cref{sec:ilp}, we present our ILP formulation to determine the vertices in the solution.

Throughout this section we assume without loss of generality that~$G$ is connected.

\subsection{Preprocessing}
\label{sec:preprocessing}
In this section we present three data reduction rules and some observations on the instance obtained after their exhaustive application.
We will also introduce the \emph{feedback edge graph}~$\widetilde G$ in this subsection, which will be used throughout the presentation of this algorithm.

Our first reduction rule deletes degree-one vertices.
This reduction rule is based on the observation that a geodetic set contains every degree-one vertex.

\begin{rrule}
	\label{rr:leaf}
	If there is a degree-one vertex $v \in V(G)$ with $N(v) = \{ u \}$, then
	\begin{itemize}
		\item decrease $k$ by 1 if $u \in T$, 
		\item add $u$ to $T$ if $u \notin T$, and
		\item delete $v$ from~$V(G)$ (and from~$T$).
	\end{itemize}
\end{rrule}

Henceforth we assume that \cref{rr:leaf} has been exhaustively applied (which can be done in linear time).
Suppose that $\fen(G) = 1$.
Then~$G$ is a cycle, and any minimal geodetic set~$S \supseteq T$ is of size at most~$|T| + 3$.
So \textsc{Extended Geodetic Set} can be solved in polynomial time when $\fen(G) \le 1$ (in fact, further analysis yields a linear-time algorithm for $\fen(G) = 1$).
We thus assume that $\fen(G) \ge 2$.

Now we introduce the \emph{feedback edge graph}~$\widetilde G$, a multigraph which is obtained from~$G$ as follows:
As long as there is a degree-two vertex $v$ with neighbors $u, w$, we remove $v$ and add an edge (multiedge) $uw$.
Using the handshake lemma, one can easily obtain the following.

\begin{observation}
	\label{obs:gtilde}
	It holds that $|V(\widetilde{G})| \le 2 \fen(G) - 2$ and $|E(\widetilde{G})| \le 3 \fen(G) - 3$.
\end{observation}
\begin{proof}
	By definition, $|E(G)| \le |V(G)| + \fen(G) - 1$.
	It follows that $|E(\widetilde{G})| \le |V(\widetilde{G})| + \fen(G) - 1$, since the number of edges decreases by~$1$ every time we remove a vertex.
	By the handshake lemma, $2|E(\widetilde{G})| = \sum_{v \in V(\widetilde{G})} \deg_{\widetilde{G}}(v) \ge 3|V(\widetilde{G})|$.
	Solving the inequalities for~$|V(\widetilde{G})|$ and $|E(\widetilde{G})|$ respectively yields the sought bounds.
\end{proof}

\begin{figure}[t]
	\centering
	\begin{tikzpicture}
		\begin{scope}[scale=.5]
		\pgfmathsetmacro{\r}{4}
		\pgfmathsetmacro{\rr}{1}
		\pgfmathsetmacro{\hone}{6}
		\pgfmathsetmacro{\htwo}{2}
		\pgfmathsetmacro{\hthree}{3}
		\pgfmathsetmacro{\hfour}{5}
		\pgfmathsetmacro{\hfive}{6}

		\node (s1) at (0, 0) [vertex, label={above left:$v_1$}] {};
		\node (s2) at (\r, 0) [vertex, label={225:$v_2$}] {};
		\node (s3) at (2, {\r*sin(60)}) [vertex, label={above:$v_3$}] {};
		\node (l1) at (-1, -.1) [vertex] {};
		\node (l2) at (-0.1, -1) [vertex] {};

		\draw (l1) -- (s1) -- (l2);
		\draw (s1) to node[vertex, pos=0.333] {} node[vertex, pos=0.666] {} (s2);
		\draw[ultra thick] (s1) -- (s2);
		\draw[bend left] (s1) to node[vertex, pos=.25] {} node[vertex, pos=.5] {} node[vertex, pos=.75] {} (s3);
		\draw[bend right] (s1) to node[vertex, pos=0.5] {} (s3);
		\draw (s2) to node[vertex, pos=0.333] {} node[vertex, pos=0.666] {} (s3);
		\draw (\r + .5, -0.3) circle (.58);
		\begin{scope}[shift={(\r + .5, -0.3)}]
			\node (sl1) at (30:0.58) [vertex] {};
			\node (sl2) at (270:0.58) [vertex] {};
			\begin{scope}[shift={(30:0.58)}]
				\node (a) at (1, 0.1) [vertex] {};
				\node (b) at (0.1, 1) [vertex] {};
			\end{scope}
			\draw (a) -- (sl1) -- (b);
		\end{scope}
		\node (l21) at (0.833, -1) [vertex] {};
		\node (l22) at (1.833, -1) [vertex] {};
		\draw (l21) -- (1.333, 0) -- (l22);
		\begin{scope}[shift={(2.666,2.31)}]	
			\node (x) at (1.1, -.15) [vertex] {};
			\node (y) at (.45, 1) [vertex] {};
			\node (y1) at (1.45, 1.1) [vertex] {};
			\node (y2) at (0.55, 2) [vertex] {};
			\draw (y1) -- (y) -- (y2);
			\draw (x) -- (0, 0) -- (y);
		\end{scope}
		\end{scope}
		\begin{scope}[shift={(5, 0)}, scale=0.5]
			\pgfmathsetmacro{\r}{4}
			\pgfmathsetmacro{\rr}{1}
			\node (s1) at (0, 0) [vertex, label={left:$v_1$}] {};
			\node (s2) at (\r, 0) [vertex, label={225:$v_2$}] {};
			\node (s3) at (2, {\r*sin(60)}) [vertex, label={above:$v_3$}] {};
			\draw[ultra thick] (s1) -- (s2) node [midway,label=below left:{}] {};
			\draw (s1) to [bend right] node [midway,below right] {} (s3);
			\draw (s1) to [bend left] node [midway,above left] {} (s3);
			\draw (\r + .5, -0.3) circle (.58);
			\draw (s2) -- (s3) node [midway,label=above right:{}] {};
			\node at (\r + 1.7 ,0) {};
		\end{scope}
	\end{tikzpicture}
	\caption{
		An illustration of an input graph $G$ (left) and $\widetilde{G}$ after \cref{rr:leaf} has been exhaustively applied (right).
		Observe that $\widetilde G$ contains no degree-one or degree-two vertex.
		For instance, a thick edge $p$ in $\widetilde{G}$ (right) corresponds to a path $P$ of length $h_p = 3$ in $G$(left).
		Moreover, we have $T_p = \{ 0, 1 \}$ after \cref{rr:leaf} has been applied exhaustively.}
	\label{fig:feg}
\end{figure}
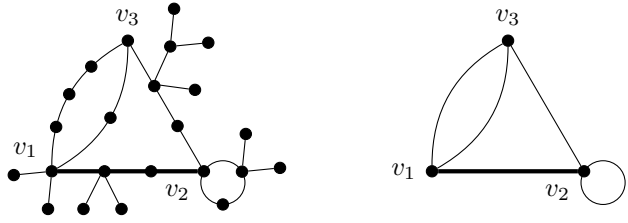

Observe that each edge~$p$ in~$\widetilde{G}$ is associated with a path $P = (p^0, p^1, \dots, p^{h_p})$ in $G$ where all of its inner vertices are of degree~2.
We sometimes refer to the endpoints $p^0, p^{h_p}$ as $p^\leftarrow, p^\rightarrow$, respectively.
Moreover, let~$T_p = \{ i \mid p^i \in T \}$ and let $p_T^\leftarrow = p^{t_p^\leftarrow}$ and $p_T^\rightarrow = p^{t_p^\rightarrow}$, where $t_p^\leftarrow = \min {T_p}$ and $t_p^\rightarrow = \max {T_p}$.
We illustrate the definitions in \cref{fig:feg}.

The following reduction rule deals with self-loops in~$\widetilde{G}$.

\begin{rrule}
	\label{rr:deg2cycle}
	If $v \in V(\widetilde{G})$ has a self-loop $p$ in $\widetilde{G}$, then decrease $k$ as follows:
	\begin{itemize}
	 	\item If $T_p = \emptyset$, then decrease $k$ by $(h_p \bmod 2)$.
	 	\item If $T_p \ne \emptyset$ and $V(P) \not\subseteq I[T_p \cup \{ v \}]$, then decrease $k$ by $|T_p|$.
	 	\item If $T_p \ne \emptyset$ and $V(P) \subseteq I[T_p \cup \{ v \}]$, then decrease $k$ by $|T_p| - 1$.
	 \end{itemize}
	Moreover, add $v$ to $T$ and remove $V(P) \setminus \{ v \}$.
\end{rrule}

\begin{lemma}
	\cref{rr:deg2cycle} is correct.
\end{lemma}
\begin{proof}
	We reduce the first two cases to the third case with the following observations:
	\begin{itemize}
		\item
			If $T_p = \emptyset$, then $(G, T, k)$ is equivalent to $(G, T \cup \{ p^{\lfloor h_p / 2 \rfloor}, p^{\lceil h_p / 2 \rceil} \}, k)$.
			Then~$V(P) \subseteq I[T_p \cup \{v\}]$ and~$|T_p| = (h_p \mod 2)$.
		\item
			If $T_p \ne \emptyset$ and $V(P) \not\subseteq I[T_p \cup \{ v \}]$, then it is equivalent either to $(G, T \cup \{p^{\lfloor h_p / 2 \rfloor}\}, k)$ or to $(G, T \cup \{p^{\lceil h_p / 2 \rceil}\}, k)$. 
			Then~$V(P) \subseteq I[T_p \cup \{v\}]$ and~$|T_p|$ increases by one.
	\end{itemize}
	So assume that $T_p \ne \emptyset$ and $V(P) \subseteq I[T_p \cup \{ v \}]$.

	Let $(G', T', k')$ be an \textsc{Extended Geodetic Set} instance as a result of \cref{rr:deg2cycle}.
	Note that $G' = G - (V(P) \setminus \{ v \})$, $T' = T \cup \{ v \}$, and $k' = k - |T_p| + 1$.
	It is easy to see that if $S \supseteq T$ is geodetic in $G$ and $|S| \le k$, then $(S \setminus V(P)) \cup \{ v \}$ is a solution of $(G', T', k')$.
	Conversely, if $S' \supseteq T'$ is a geodetic set in $G'$ of size at most $k'$, then $(S' \setminus \{ v \}) \cup T_p$ is a solution of $(G, T, k)$.
\end{proof}

The next reduction rule ensures that for every~$p \in E(\widetilde G)$ with $T_p \ne \emptyset$, there is a shortest path from an endpoint of $P$ to the closest vertex in~$T_p$ that is contained inside~$P$.
For this we introduce the following notation.
Let $\mathcal{R} = \{ \leftarrow, \rightarrow \}$.
For $r \in \mathcal{R}$, we denote by~$\overline{r} \in \mathcal R \setminus \{r\}$ the opposite direction.

\begin{rrule}
	\label{rr:addleaf}
	Let $p \in E(\widetilde{G})$ with $T_p \ne \emptyset$, and let~$r \in \mathcal R$.
	If~$d_P(p_T^r, p^r) > d_P(p_T^r, p^{\overline r}) + d_G(p^{\overline r}, p^r)$,
	then add~$p'$ to~$T$, where~$p'$ is between~$p_T^r$ and~$p^r$ and
	$d(p', p_T^r) = \lfloor (h_p + d_G(p^\leftarrow, p^\rightarrow)) / 2 \rfloor$.
\end{rrule}
\begin{lemma}
	\cref{rr:addleaf} is correct.
\end{lemma}
\begin{proof}
	Suppose that $(G, T, k)$ is a yes-instance with a solution $S \supseteq T$.
	Let $P'$ be a subpath of $P$ with endpoints $p'$ and $p_T^r$.
	Note that
	\begin{align*}
		d_{P}(p', p^r) + d_G(p^r, p^{\overline{r}}) + d_{P}(p^{\overline{r}}, p_T^r)
		&= (t_P^r - d_P(p', p_T^r)) + d_G(p^r, p^{\overline{r}}) + (h_p - t_P^r) \\
		&= \lceil (h_i + d_G(p^r, p^{\overline{r}})) / 2 \rceil.
		% \ge d_{P}(p', p_T^r).
	\end{align*}
	Thus, $S$ must contain a vertex $v \in V(P) \setminus \{ p_T^r \}$ to cover $P'$.
	The correctness follow, because $(S \setminus \{ v \}) \cup \{ p' \}$ is also geodetic in $G$.
\end{proof}

\subsection{Guessing}
\label{sec:guessing}
Towards obtaining a geodetic set~$S$ of size at most~$k$, we extend our current set~$T$ of vertices fixed in the solution.
First we guess the set of endpoints that are in the solution.
Next, using another reduction rule, we fix further vertices that are required to be in the geodetic set of our interest.
These vertices possibly depend on the (previously guessed) endpoints that are in the solution.
Finally, we guess how many vertices we need to add to every path~$P$ for~$p \in E(\widetilde G)$.
Then, the exact positions of these vertices are determined using ILP.

Suppose that~$(G, T, k)$ is a yes-instance.
We fix a solution~$S$ of minimum size that maximizes the number~$|S \cap V(\widetilde G)|$ of endpoints among all such solutions.
Intuitively, our goal is to find $S$.
To do so, we first guess the set~$\widetilde S = S \cap V(\widetilde G)$ of endpoints in~$S$;
there are at most~$2^{|V(\widetilde G)|} \le 2^{2\fen(G) - 2}$ possibilities by \cref{obs:gtilde}.
We extend $T$ by adding all vertices from $\widetilde S$.
So we will henceforth assume that $S \cap V(\widetilde G) = T \cap V(\widetilde G)$.
Using another reduction rule, we ensure that for every~$p \in E(\widetilde G)$, the vertices between~$p_T^\leftarrow$ and~$p_T^\rightarrow$ are covered.
\begin{rrule}
	\label{rr:deg1deg2deg1}
	Let $p \in E(\widetilde{G})$. 
	If there are $t < t' \in T_p$ such that $[t + 1, t' - 1] \cap T_p = \emptyset$ and $d_G(p^t, p^{t'}) < t' - t$ (equivalently, $d_G(p^\leftarrow, p^\rightarrow) + h_p < 2t' - 2t$),
	then add $\lfloor (t + t') / 2 \rfloor$ to $T$.
\end{rrule}
\begin{lemma}
	\cref{rr:deg1deg2deg1} is correct.
\end{lemma}
\begin{proof}
	Let $S$ be a geodetic set and let $V_p = \{ p^{t + 1}, \dots, p^{t' - 1} \}$.
	It suffices to show that $S \cap V_p \ne \emptyset$ and that $S' = (S \setminus V_p) \cup \{ p^{\lfloor (t + t') / 2 \rfloor} \}$ is geodetic.
	Suppose that $S \cap V_p = \emptyset$.
	Then for each $v, v' \in S$, no shortest path between $v$ and $v'$ visits a vertex in $V_i$.
	Hence, we have $S \cap V_p \ne \emptyset$.
	For the latter part, it is easy to see that $S'$ is geodetic because $V_p \subseteq I[\{ p^t, p^{\lfloor (t + t') / 2 \rfloor}, p^{t'} \}]$.
\end{proof}

We will prove two lemmata required for the next guessing step and for the subsequent ILP formulation.
First, we show that $S$ contains no vertex on a path~$P$ for~$p \in E(\widetilde G)$ with~$T_p \ne \emptyset$.

\begin{lemma}
	\label{lemma:nonemptyp}
	Let $p \in E(\widetilde{G})$ with $T_p \ne \emptyset$.
	Then, $S \cap V(P) \subseteq T_p$.
\end{lemma}
\begin{proof}
	For $r \in \mathcal{R}$, suppose that~$S$ contains a vertex $p^i \in V(P) \setminus T_p$ that lies between $p^r$ and~$p^r_T$.
	Since \cref{rr:addleaf} is applied exhaustively, $(S \setminus \{ p^i \}) \cup \{ p^r \}$ is also a solution of minimum size, contradicting the maximality of $|S \cap V(\widetilde G)|$.
	Thus, it remains to show that~$S$ contains no vertex that lies between $p_T^\leftarrow$ and $p_T^\rightarrow$ in $P$.
	Note that after applying \cref{rr:deg1deg2deg1}, each vertex in~$P$ between~$p_T^\leftarrow$ and~$p_T^\rightarrow$ are included in~$I[T_p]$.
	Due to its minimality, $S$ contains no vertex $p^i \in V(P) \setminus T_p$ between $p^\leftarrow_T$ and $p^{\rightarrow}_T$ in~$P$.
\end{proof}

We also show that $S$ contains at most two inner vertices of~$P$ if~$T_p = \emptyset$ for~$p \in E(\widetilde G)$.
\begin{lemma}
	\label{lemma:emptyp}
	Let $p \in E(\widetilde G)$ with $T_p = \emptyset$.
	Then, $|S \cap V(P)| \le 2$.
\end{lemma}
\begin{proof}
	If~$|S \cap V(P)| = 3$, then $(S \setminus V(P)) \cup \{ p^\leftarrow, p^{\lfloor h_p/2\rfloor}, p^\rightarrow \}$ is also a minimum solution, contradicting the fact that $|S \cap V(\widetilde G)|$ is maximized.
\end{proof}

Now we make further guesses.
For each edge~$p \in E(\widetilde G)$, we guess the number $n_p \in \{0, 1, 2\}$ of inner vertices in $S \cap V(P)$.
Note that there are at most~$3^{|E(\widetilde G)|} \le 3^{3\fen(G)-3}$ possibilities by \cref{obs:gtilde}.
The next step is to determine exactly which vertices to take using ILP.

\subsection{Finding a minimum geodetic set via ILP}
\label[section]{sec:ilp}
\newcommand{\oI}{E}
\newcommand{\ou}{p_S}
Let~$\oI_n = \{p \in E(\widetilde G) \mid T_p=\emptyset, n_p = n\}$ for~$n \in \{0, 1, 2\}$ and let~$\oI' = \{p \in E(\widetilde G) \mid T_p \ne \emptyset\}$.
Further, let~$\mathcal \oI = \oI_1 \cup \oI_2 \cup \oI' = E(\widetilde G) \setminus \oI_0$.
Note that~$S$ contains at least one vertex in~$V(P)$ for every~$p \in \mathcal \oI$.
For each~$p \in \mathcal \oI$, we introduce two nonnegative variables~$x_p^\leftarrow, x_p^\rightarrow$, and let~$\ou^\leftarrow = p^{x_p^\leftarrow}$ and~$\ou^\rightarrow = p^{h_p - x_p^\rightarrow}$.
The intended meaning of~$x_p^\leftarrow$, respectively~$x_p^\rightarrow$ is that~$S$ contains~$\ou^\leftarrow$, respectively~$\ou^\rightarrow$.
Then the geodetic set of our interest will be given by~$X = T \cup \bigcup_{p \in \oI_1 \cup \oI_2} \{\ou^\leftarrow, \ou^\rightarrow\}$.
For each~$p \in \mathcal \oI$ we add the following constraints:
\begin{equation} \label[constraint]{constraint:x}
  \begin{cases}
    x_p^\leftarrow > 0, x_p^\rightarrow > 0, \text{and } x_p^\leftarrow + x_p^\rightarrow \le h_p & \text{if } p \in \oI_1 \cup \oI_2, \\
    x_p^\leftarrow + x_p^\rightarrow = h_p & \text{if } p \in \oI_1, \\
    h_p - 2x_p^\leftarrow - 2x_p^\rightarrow \le d_G(v_p^\leftarrow, v_p^\rightarrow) & \text{if } p \in \oI_2, \\
    x_p^\leftarrow = p_T^\leftarrow \text{ and } x_p^\rightarrow = h_p - p_T^\rightarrow & \text{if } p \in \oI'.
  \end{cases}
\end{equation}

Let $V_p^\leftarrow = \{ p^1, \dots, p^{x_{p}^\leftarrow - 1} \}$ and $V_p^\rightarrow = \{ p^{h_p - x_i^\rightarrow + 1}, \dots, p^{h_p - 1} \}$ for each $p \in \mathcal{\oI}$.
We show that \cref{constraint:x} guarantees that the vertices between $\ou^\leftarrow$ and $\ou^\rightarrow$ are covered if $p \not\in \oI_0$.

\begin{lemma}
	\label{lemma:constraintx}
	If \cref{constraint:x} is fulfilled, then $Q_p = V(P) \setminus (\{ p^\leftarrow, p^\rightarrow \} \cup V_p^\leftarrow \cup V_p^\rightarrow) \subseteq I[S]$ holds for each $p \in \mathcal{\oI}$.
\end{lemma}
\begin{proof}
If $p \in \oI_1$, then we have~$Q_p = \{ p^{x_p^\leftarrow} \} = \{ p^{x_p^\rightarrow} \}$ and hence $Q_p \subseteq I[S]$.
	If $p \in \oI_2$, then we have~$Q_p = \{ p^{x_p^\leftarrow}, p^{x_p^\leftarrow + 1}, \dots, p^{x_p^\rightarrow} \}$.
	It follows from \cref{constraint:x} that $d_{P}(\ou^\leftarrow, \ou^\rightarrow) \le d_{P}(\ou^\leftarrow, p^\leftarrow) \allowbreak + d_G(p^\leftarrow, p^\rightarrow) + d_{P}(p^\rightarrow, \ou^\rightarrow)$.
	This implies that $Q_p \subseteq I[\ou^\leftarrow, \ou^\rightarrow] \subseteq I[S]$.
	Finally, if~$p \in \oI'$, then all vertices in $Q_p$ are covered as shown in \cref{lemma:nonemptyp}.
\end{proof}

Next, we introduce constraints to determine whether there is a shortest path between~$\ou^r$ and $q_S^{s}$ visiting $p^r$ and $q^s$, for each $p \ne q \in E(\widetilde G)$ and $r, s \in \mathcal{R}$ (recall that $\mathcal R = \{\leftarrow, \rightarrow\}$).
Using binary variables~$a_{p, q}^{r, s}, b_{p, q}^{r, s}, c_{p, q}^{r, s},\allowbreak z_{p, q}^{r, s}$, we add the following constraints for each $p \ne q \in \mathcal{\oI}$ and $r, s \in \mathcal{R}$.
Informally, if~$z_{p, q}^{r, s} = 1$, then there exists a shortest path as described above.
\begin{equation} \label[constraint]{constraint:nzerocover}
	\left\{
	\begin{aligned}
		&(x_p^r + d_G(p^r, q^s) + x_q^s) - (x_p^r + d_G(p^r, q^{\overline{s}}) + h_q - x_q^s)
		\le N(1 - a_{p, q}^{r,s}), \\
		&(x_p^r + d_G(p^r, q^s) + x_q^s) - (h_p - x_p^r + d_G(p^{\overline{r}}, q^s) + x_q^s)
		\le N(1 - b_{p, q}^{r, s}),\\
		&(x_p^r + d_G(p^r, q^s) + x_q^s) - (h_p - x_p^r + d_G(p^{\overline{r}}, q^{\overline{s}}) + h_q - x_q^s)
		\le N(1 - c_{p, q}^{r, s}),\\
		& 3 - a_{p, q}^{r, s} - b_{p, q}^{r, s} - c_{p, q}^{r, s} \le 3 - 3 z_{p, q}^{r, s}.
	\end{aligned}
	\right.
\end{equation}
Here~$N$ is some sufficiently large number (i.e., $N = 100 \cdot |E(G)|$ will do).
\begin{lemma}
	\label{lemma:middlecover}
	If \cref{constraint:nzerocover} is fulfilled with $z_{p, q}^{r, s}=1$, then $I[p^{r}, q^{s}] \subseteq I[p_S^{r}, q_S^{s}]$.
\end{lemma}
\begin{proof}
	We fix~$p, q, r, s$ and remove them from the sub- and superscripts of the binary variables.
	If~$z = 1$, then we obtain~$3 - a - b - c \le 0$.
	Since~$a, b, c \in \{0, 1\}$ we have that~$a = b = c = 1$, which in turn implies that there is a shortest path between~$u_p^r$ and~$u_q^s$ that visits~$p^r$ and~$q^s$.
\end{proof}

We add a similar constraint for shortest paths between~$\ou^\leftarrow$ and~$\ou^\rightarrow$ for each~$p \in E(\widetilde G)$.
For each~$p \in \mathcal \oI$ and~$r \in \mathcal R$ we add the constraint
\begin{equation}
	\label[constraint]{constraint:nzerocover2}
	(x_p^r + d_G(p^r, p^{\overline{r}}) + x_p^{\overline{r}}) - (h_p - x_p^r - x_p^{\overline{r}}) \le N(1 - z_{p, p}^{r, \overline{r}}).
\end{equation}
Here $z_{p, p}^{r, \overline{r}}$ is a binary variable.
It is easy to see that if $z_{p, p}^{r, \overline{r}} = 1$, then there is a shortest path from~$\ou^r$ to $\ou^{\overline{r}}$ going through $p^r$ and $p^{\overline{r}}$.

Now we use \cref{constraint:nzerocover,constraint:nzerocover2} to cover the remaining vertices.
First we handle the paths without any solution vertex.
For each $\ell \in \oI_0$, we add the following constraint to guarantee that there are $p, q \in E(\widetilde G)$ and $r, s \in \mathcal{R}$ such that $V(L) \subseteq I[\ou^r, q_S^{s}]$, where~$L$ is the path associated with~$\ell$:
\begin{equation}
	\label[constraint]{constraint:izerocover}
	\sum_{
		\substack{
			p, q \in \mathcal{\oI}, \, r, s \in \mathcal{R}, (p, r) \ne (q, s) \\
			d(p^r, \ell^{\leftarrow}) + h_\ell + d(\ell^{\rightarrow}, q^s) = d(p^r, q^s)
		}
	}
	z_{p, q}^{r, s}
	\ge 1.
\end{equation}
To ensure that every vertex~$v \in V(\widetilde G) \setminus \widetilde S$ is covered, we add \cref{constraint:izerocover}, where~$L$ is a path of length zero with endpoint~$v$, that is, $h_\ell = 0$ and~$\ell^\leftarrow = \ell^\rightarrow = v$.

Finally, we deal with the vertices in $V_p^\leftarrow$ and $V_p^\rightarrow$.
Note that for each $p \in E(\widetilde G)$ and $r \in \mathcal{R}$, the vertices in $V_p^r$ are covered if
\begin{itemize}
	\item
		it holds that $x_p^r \le 1$ (that is, $V_p^r = \emptyset$), or
	\item
		there is $q \in E(\widetilde G)$ and $s \in \mathcal{R}$ such that a shortest $\ou^r$--$q_S^s$-path visits~$p^r$.
\end{itemize}
For each $p \in \mathcal{\oI}$ and $r \in \mathcal{R}$, let $y_p^r$ be a binary variable and add the following constraint:
\begin{equation}
	\label[constraint]{constraint:icover}
	\begin{aligned}
		x_p^r - 1 \le N(1 - y_p^r) \quad\text{and}\quad
		y_p^r + \sum_{q \in E(\widetilde G), s \in \mathcal{R}} z_{p, q}^{r, s} \ge 1.
	\end{aligned}
\end{equation}
It is easy to verify that if $y_p^r = 1$, then $x_p^r \le 1$ must hold.
This concludes the ILP formulation.
We show that our ILP formulation finds a minimum geodetic set.

\begin{theorem}
	\label{thm:fen}
	\emph{\textsc{Geodetic Set}} can be solved in~$O^*(2^{O(\fen(G)^2)})$ time.\footnote{the~$O^*(\cdot)$ notation hides factors that are polynomial in the input size}
\end{theorem}
\begin{proof}
	We prove that there is a geodetic set~$S \supseteq T$ satisfying \cref{lemma:nonemptyp,lemma:emptyp} if and only if one of our ILP instances is a yes-instance.
	The forward direction is clearly correct.
	The correctness of the other direction is due to the following observations.
	\begin{itemize}
		\item
			The vertices in~$P$ for~$p \in \oI_0$ as well as the vertices in~$V(\widetilde G) \setminus \widetilde S$ are covered because of \cref{constraint:izerocover}.
		\item
			For each $p \in \mathcal{\oI}$, $V_i^\leftarrow$ and $V_i^\rightarrow$ are covered due to \cref{constraint:icover}.
			The remaining vertices are covered due to \cref{lemma:middlecover}.
	\end{itemize}
	Note that we construct $2^{O(\fen(G))}$ instances of ILP.
	Each ILP instance uses $O(\fen(G)^2)$ binary variables and $O(\fen(G))$ variables which are not necessarily binary.
	To solve one ILP instance, we first try every assignment to binary variables (note that there are $2^{O(\fen(G)^2)}$ assignments).
	Then, we solve an ILP instance with $O(\fen(G))$ variables, which requires $O^*(\fen(G)^{O(\fen(G))})$ time \cite{Len83}.
	This results in an algorithm whose running time is~$O^*(2^{O(\fen(G)^2)})$.
\end{proof}

\section{Fixed-Parameter Tractability for Clique-Width with Diameter}
\label{sec:cw}
In this section we obtain fixed-parameter tractability results for clique-width combined with diameter, and for tree-depth.
Our algorithm is based on a theorem by Courcelle et al.~\cite{CMR00}:
If a graph property~$\pi$ can be expressed as a formula $\varphi$ in~$\mso_1$ logic, then whether a graph $G$ has~$\pi$ can be determined in $O(f(\cw(G) + |\varphi|) \cdot (|V(G)| + |E(G)|))$ time for some function $f$.

\begin{theorem}
	\label{cor:twdiam}
	\emph{\textsc{Geodetic Set}} is fixed-parameter tractable with respect to $\cw(G) + \diam(G)$.
\end{theorem}
\begin{proof}
	We describe how to express \textsc{Geodetic Set} in $\mso_1$ logic.
	We define
	\begin{align*}
		\varphi = \exists S \: \left(\forall v \: [\exists u, w \: (u \in S \wedge w \in S \wedge \visit(u, v, w))]\right),
	\end{align*}
	where $\visit(u, v, w)$ is true if and only if there is a shortest path $u$--$w$ visiting $v$.
	It remains to construct $\visit(u, v, w)$.
	First, let us define a formula $\spath(v_1, \dots, v_i)$ which evaluates to true if and only if $(v_1, \dots, v_i)$ is a path:
	\begin{align*}
		\spath(v_1, \dots, v_\delta) = \bigwedge_{j \in [i - 1]} v_j v_{j + 1} \in E(G).
	\end{align*}
	We then define $\idist_i(u, w)$ which is true if and only if $d_G(u, w) = i$.
	\begin{align*}
		\idist_i(u, w) =& \exists v_2, \dots, v_{i - 1} \, (\spath(u, v_2, \dots, v_{i - 1}, w)) \\
		&\wedge \bigwedge_{j \in [i - 1]} \nexists v_2, \dots, v_{j - 1} (\spath(u, v_2, \dots, v_{j - 1}, w)).
	\end{align*}
	Finally, we define $\visit(u, v, w)$:
	\begin{align*}
		\visit(u, v, w) = \bigvee_{i \in [\diam(G)]} \left( \idist_i (u, w) \wedge \left[ \bigvee_{j \in [i - 1]} \idist_j (u, v) \wedge \idist_{j - i} (v, w) \right] \right).
	\end{align*}
	Note that $|\varphi| \in \diam(G)^{O(1)}$.
	Thus, fixed-parameter tractability for $\cw(G) + \diam(G)$ follows from Courcelle's theorem.
\end{proof}

Note that $\cw(G) \le 2$ and $\diam(G) \le 2$ for any cograph $G$.
Thus, our result extends polynomial-time solvability on cographs proven by Dourado et al.~\cite{DPRS10}.

We also obtain fixed-parameter tractability for tree-depth as well as for modular-width from \cref{cor:twdiam}.
The tree-depth of a graph~$G$ can be roughly approximated by~$\log h \le \td(G) \le h$, where~$h$ is the height of a depth-first search tree of~$G$ \cite{NM12}.
Hence, the length of all paths in~$G$, specifically the diameter of~$G$, is at most~$2^{\td(G)}$.
Moreover, $\cw(G) \le 3 \cdot 2^{\tw(G) - 1}$ \cite{CR05} and $\tw(G) \le \td(G) - 1$.
Similarly, $\cw(G) \le \mw(G)$ (by definition) and~$\diam(G) \le \max \{2, \mw(G)\}$ \cite{KLMP18}.
Consequently, we obtain the following.

\begin{corollary}
	\label{thm:tdtd}
	\emph{\textsc{Geodetic Set}} is fixed-parameter tractable with respect to tree-depth and with respect to modular-width.
\end{corollary}

\section{Conclusion}
We initiated a parameterized complexity study of \textsc{Geodetic Set} for parameters measuring tree-likeness.
We conclude this work by suggesting some future research directions.
None of the fixed-parameter algorithms presented in this work are practical.
Are there more efficient fixed-parameter algorithms with respect to feedback edge number, tree-depth or modular-width?
Further, while we can quite surely exclude fixed-parameter tractability for feedback vertex number and path-width, it is still open whether \textsc{Geodetic Set} is in XP with any (combination) of these parameters.
Recall that the related \textsc{Geodetic Hull} problem is in XP with respect to tree-width \cite{KMS19}, but for \textsc{Geodetic Set}, even the complexity on series-parallel graphs (which have tree-width two) is unknown.

Going to related problems and parameters, it is open whether \textsc{Metric Dimension} is fixed-parameter tractable with respect to the feedback edge number \cite{Epp15}.
This is especially interesting since the problem behaves similarly to \textsc{Geodetic Set} in terms of complexity:
\textsc{Metric Dimension} is fixed-parameter tractable with respect to tree-depth \cite{San17} and with respect to modular-width \cite{BFGR17}, but~$\Wone$-hard with respect to path-width \cite{BP19} and~$\Wtwo$-hard with respect to the solution size \cite{HN13}.
We are optimistic that the method presented in \cref{sec:fen} can be used to answer this question positively, especially since Epstein et al.~\cite{ELW15} showed that the number of solution vertices on a path of degree-two vertices (cf.\ \cref{lemma:emptyp}) is bounded by a constant.

\paragraph{Acknowledgement.}
We thank Lucia Draque Penso (Ulm University) for suggesting studying \textsc{Geo\-detic Set} from a view of parameterized complexity, and we thank André Nichterlein and Rolf Niedermeier (both TU Berlin) for helpful feedback and discussion.
We are also grateful to an anonymous reviewer for suggesting that the ILP instances in \Cref{sec:fen} can be solved more efficiently.

\bibliographystyle{plainurl}
\bibliography{geodetic-set}

\end{document}